%% file: ants.tex
\documentclass[11pt]{article}
\usepackage[margin=0.6in]{geometry}

\usepackage{authblk}
\usepackage{hyperref}

\usepackage{amsthm}
\theoremstyle{plain}
\newtheorem{theorem}{Theorem}[section]

\newtheorem{corollary}[theorem]{Corollary}
\theoremstyle{definition}
\newtheorem{definition}[theorem]{Definition}
\newtheorem{observation}[theorem]{Observation}
\newtheorem{remark}{Remark}

\usepackage{algorithm}
\usepackage{algpseudocode}
\algnewcommand{\IIf}[1]{\State\algorithmicif\ #1\ \algorithmicthen}
\algnewcommand{\EndIIf}{\unskip\ \algorithmicend\ \algorithmicif}

\usepackage{graphicx}
\usepackage{tikz}

\begin{document}
\pagenumbering{arabic}

\title{Optimal and Resilient Pheromone Utilization in Ant Foraging}

\author{Yehuda Afek}
\author{Roman Kecher}
\author{Moshe Sulamy}
\affil{ \small{
\{afek, roman.kecher, moshe.sulamy\}@cs.tau.ac.il} \\
Tel-Aviv University}
\date{\today}
\maketitle

\begin{abstract}
\input{00-abstract.tex}
\end{abstract}



\section{Introduction}
\input{01-introduction.tex}

\section{Model}
\label{section:model}
\input{02-model.tex}

\section{The case of FSM Ants}
\label{section:fsm}
\input{03-fsm.tex}

\section{The case of Fault-Tolerant Ants}
\label{section:ft}
\input{04-ft.tex}

\section{The case of TM Ants}
\label{section:tm}
\input{05-tm.tex}

\section{Summary}
\label{section:summary}
\input{06-summary.tex}

\clearpage

\bibliographystyle{abbrv}
\bibliography{ants}

\end{document}

%% file: 00-abstract.tex

Pheromones are a chemical substance produced and released by ants as 
means of communication.
In this work we present the minimum amount of pheromones necessary and sufficient for a
colony of ants (identical mobile agents) to deterministically find a food source (treasure),
assuming that each ant has the computational capabilities of either a Finite State Machine (FSM) or a
Turing Machine (TM). 
In addition, we provide pheromone-based foraging algorithms capable of handling 
fail-stop faults.

In more detail, we consider the case where $k$ identical ants, initially located at the center (nest) of an
infinite two-dimensional grid and communicate only through pheromones,
perform a collaborative search for an adversarially hidden treasure placed at an unknown distance $D$. 
We begin by proving a tight lower bound of $\Omega(D)$ on the amount of pheromones
required by any number of FSM based ants to complete the search,
and continue to reduce the lower bound to
$\Omega(k)$ for the stronger ants modeled as TM.
We provide algorithms which match the aforementioned lower bounds, and still terminate in
optimal $\mathcal{O}(D + D^2 / k)$ time, under both the synchronous and asynchronous
models.
Furthermore, we consider a more realistic setting, where an unknown number $f < k$ of
ants may fail-stop at any time; we provide fault-tolerant FSM algorithms (synchronous and
asynchronous), that terminate in $\mathcal{O}(D + D^2/(k-f) + Df)$ rounds and emit no more
than the same asymptotic minimum number of $\mathcal{O}(D)$ pheromones overall.

%% file: 01-introduction.tex

While performing mostly simple and local computations, ants solve complicated problems
 in collaboration.
Such collaborative work, with limited communication and computational power, is at the heart of distributed computing.
By studying and understanding the behavior of ants from a distributed computing perspective, progress and insights can be made for both distributed computing and biology.

One task that is successfully solved by ants on a daily basis is that of foraging;
this is the basis for the \emph{Ants Nearby Treasure Search (ANTS)} problem, first presented by Feinerman, Korman, Lotker and Sereni \cite{Feinerman2012Collaborative}.
In the original ANTS problem, $k$ identical ants start at the origin (i.e., the \emph{nest}) of
 an infinite grid, and their goal is to find a treasure (i.e., the \emph{food}) which is 
located at an unknown (Manhattan)
 distance
$D \in \mathcal{N}$ from the nest. The ants are modeled as randomized mobile Turing Machines (henceforth TM), traveling along $\mathcal{Z}^2$, and it is assumed that they cannot communicate after
leaving the nest.
Once any ant steps onto the grid point containing the food, the food is found and the algorithm terminates.
The lower bound on the time required to find the food, $\Omega(D+D^2/k)$, is also 
noted.

However, it is known that real ants have limited means of communication; for instance, ants are able to
 produce, emit and sense pheromones \cite{Jackson2006Communication}.
This work considers a variant of the aforementioned ANTS problem, in which ants are
able to communicate only through pheromones \cite{Lenzen2013Pheromones}. We present a tight lower
bound of $\Omega(D)$ on the number of pheromones necessary by $k$ deterministic mobile Finite State
Machines (henceforth FSM) to guarantee a successful search, and show that this amount
of pheromones is in fact enough to complete the search in
optimal $\mathcal{O}(D + D^2 / k)$ time. Moreover, fault-tolerant synchronous and
asynchronous FSM algorithms are provided for the case where an unknown number
$f < k$ of the ants may fail-stop at any time; these algorithms terminate in 
$\mathcal{O}(D + D^2 / (k-f) + D f)$ time, whilst only consuming the same number of 
$\mathcal{O}(D)$ pheromones. In an attempt to reduce the number of required
 pheromones, stronger ants modeled as TM are considered. While TM ants do not 
have to use pheromones to conduct a thorough search, it is shown that $\Omega(k)$
 pheromones are required in order to achieve an optimal time search.

\paragraph{Outline}
The rest of this document is organized as follows.
The current section proceeds with a summary of related and previous work on the
 ANTS problem, and concludes with a concise  summary of the contributions of this research.
Section \ref{section:model} covers the models in use:
it describes the general model of ANTS with pheromones, the complexity measures 
employed throughout this work, and the synchronous and asynchronous scheduling models.
Section \ref{section:fsm} presents tight lower bounds on the amount of pheromones
required by FSM ants (synchronous and asynchronous), and section \ref{section:ft} 
provides fault-tolerant FSM algorithms.
Section \ref{section:tm} reduces the lower bound for TM ants, while still maintaining an
optimal run-time. Section \ref{section:summary} summarizes this work and its
results.

\section{Related work}

Feinerman, Korman, Lotker and Sereni have introduced the original ANTS problem in
\cite{Feinerman2012Collaborative}. They have shown that if each ant is a randomized TM
which cannot communicate after leaving the nest, then
matching the optimal runtime of $\mathcal{O}(D + D^2 / k)$ requires that ants have
 knowledge or a constant approximation of $k$, their total number.
In a later work, Feinerman and Korman \cite{Feinerman2012Memory} have also
analyzed the trade-offs between the exact memory size available to each individual agent (ant)
 and the best achievable run-time.

Another line of work, presented in \cite{Emek2013FSM, Emek2014FSM, Emek2014HowMany, Langner2014FaultTolerant}, 
assumes that when ants are sufficiently close to each other, they are able to communicate through 
messages of constant size; in \cite{Emek2013FSM, Emek2014FSM} Emek et al. model
ants as mobile (probabilistic) FSM with said communication capabilities, and
 show that the
optimal run-time can be achieved in this model as well. The successful termination of these
distributed algorithms relies on the communication patterns of the ants; therefore, Emek
et al. also discuss the minimum number of ants that are required \cite{Emek2014HowMany}.
Finally, a fault-tolerant FSM algorithm is presented by Langner et al. in
 \cite{Langner2014FaultTolerant}, but only for the synchronous case. This algorithm
assumes that an unknown number $f \le n/13$ of the ants may fail-stop, and it is shown that it
terminates in time $\mathcal{O}(D + D^2 / k + D f)$ with high probability.

A more realistic model of communication is based on the capability of ants 
 to use chemical signals called pheromones, that can be emitted by one ant and sensed
by another \cite{Jackson2006Communication}. 
A corresponding model of the ANTS problem, which supports communication through pheromones,
has been introduced by Lenzen and Radeva \cite{Lenzen2013Pheromones}. In that model,
each ant is a deterministic mobile FSM, that is able to
emit and sense pheromones. It is assumed that ants emit a pheromone on every explored
spot, and so an algorithm that uses $\mathcal{O}(D^2)$ pheromones and 
terminates in optimal $\mathcal{O}(D+D^2/k)$ time is presented.

Pheromones are a biological resource which might be limited in quantity or costly to use
in nature. Therefore, in our model \cite{AKS2014Pheromones}
ants are able to \emph{choose} whether or not to emit a pheromone upon exploring a grid
 point. Thus, analyzing the minimum amount of pheromones that is required in order to
 find the food, 
under various computational models, is a natural research direction.

The use of pheromones or similar chemical signals during the foraging process of real 
biologic ants is still not entirely clear, and constitutes an active research area in biology.
Some reports suggest that ants
use pheromones only to mark the trail to a food item that has been already found, while
other observations support the claim that pheromones do, in fact, control the foraging
 direction \cite{Greene2007Patrollers, Gordon2010Ants}. In this work we consider this
question from a theoretic distributed computing perspective, but still hope that 
it may help with understanding real ants.

\section{Contributions}

The contributions of this work are as follows:

\begin{description}
\item[Model] The model of Lenzen and Radeva \cite{Lenzen2013Pheromones} 
is extended in a natural direction, by allowing ants to choose whether to emit
pheromones or not. Communication through pheromones is modeled carefully
for both synchronous and asynchronous settings.
\item[Lower bounds] It is shown that FSM ants must use $\Omega(D)$ pheromones (where $D$ is
the distance to the food) to complete the search, and TM ants must use $\Omega(k)$ pheromones (where $k$ is the total number of ants)
 to find the treasure in an optimal $\mathcal{O}(D + D^2/k)$ amount of time.
\item[Upper bounds] Four different algorithms are provided, for each combination of FSM/TM ants
and synchronous/asynchronous settings. These algorithms terminate in optimal time and
require $\mathcal{O}(D)$ pheromones for FSM and $\mathcal{O}(k)$  for TM ants, proving the
lower bounds to be tight.
\item[Fault-tolerance] Another set of two (synchronous and asynchronous FSM) algorithms is
 provided for the case where $f < k$ of the ants may
 fail-stop at any time. These algorithms terminate in $\mathcal{O}(D + D^2/(k-f) + Df)$ time, whilst
 only requiring the same number of $\mathcal{O}(D)$ pheromones.
\end{description}

One interesting quality of the algorithms presented in this work, is that they do
not differ that much; we have made a big effort to keep the algorithms as similar as possible,
 and have constructed one on top of the other in an incremental manner wherever possible.

%% file: 02-model.tex

We consider the ANTS problem in which $k \in \mathcal{N}$ identical mobile agents (ants) are traveling
along the discrete two-dimensional grid $\mathcal{Z}^2$, foraging for food. 
All ants start at $(0, 0) \in \mathcal{Z}^2$
the center of the grid, a point which is also referred to as the nest. A single food source (treasure) is
 located at some distance $D \in \mathcal{N}$ from the nest, such that the position of the
 treasure and
the distance to it are unknown. Manhattan distance is used, such that a treasure at 
$(i, j) \in \mathcal{Z}^2$ is said to be at distance $D = |i| + |j|$. 
Each ant executes the same algorithm until the treasure is found, at which point
 the distributed search process terminates successfully.
We only consider \emph{uniform algorithms},
i.e. every single algorithm should work for all $k, D \in \mathcal{N}$.

The only way in which ants may communicate is through pheromones; 
under the proposed model even two ants sharing
the same exact location are not aware of this fact.
Every ant is able to emit
a pheromone in its current location, as well as sense a pheromone previously emitted at
that spot by others (including itself). An ant can choose whether to emit a pheromone, or not.
There is only a single type of pheromone, and once a pheromone is
emitted it will stay at that grid point forever. We assume that the quantity of
the pheromones is unlimited, but discuss the required amount throughout this paper.

We start by considering deterministic FSM ants with a constant size memory, and proceed
to consider stronger deterministic TM ants with unbounded amount of memory.
However, even in the case of TM ants, we actually only require $\mathcal{O}(\log D)$ space\footnote{Assuming $D > k$, which will be a prominent assumption throughout this work.}.

\paragraph{Steps} Execution advances with ants performing \emph{step}s. A single, atomic step of an ant
consists of the following sequence: 
(1) sense if a pheromone exists in the current location,
(2) use all local information (current state, memory contents and pheromone sensing) 
to decide
on the current actions (pheromone emission and move direction\footnote{Four move directions are labeled as North, East, South and West. Holding current position is also possible.}) and state transition, and
(3) execute the optional pheromone emission and then the move,
and update state. 
The arrival to the next grid point and
the repeated execution from (1) are considered to be part of the next step.

Both synchronous and asynchronous scheduling models are considered. In the asynchronous
case, an adversary scheduler repeatedly selects a single ant to perform one step. It follows
 that in the asynchronous case, sensing and emitting a pheromone on the current location of
an ant is atomic, and is carried out entirely in a single step.

In the synchronous scheduling model, all the
active ants perform a single step together in each round of the algorithm. Therefore, 
 two ants standing on the same spot will both notice the
existence, or lack thereof, of a pheromone at the beginning of their steps, and might
make the same decisions from that point onwards. They will not have any means
to break the symmetry in this case. Due to this issue, in the synchronous case we will assume a gradual
 release of the ants from the nest at $(0,0)$, such that no two ants leave the nest at the same 
time\footnote{Mechanisms that achieve this are referred to as \emph{emission schemes}\cite{Emek2013FSM, Emek2014FSM}, and are beyond the scope of this work.}.
Our analysis of the expected run time (which is measured in \emph{rounds}, defined below)
will assume this "nest emission" time to be negligible compared to the total search time.

\paragraph{Complexity} The running-time of the proposed algorithms is measured in \emph{rounds}.
The definition is the same for both scheduling models: a round ends as soon as every ant
took at least a single step. Note that this definition is not equivalent to defining the number
of rounds to be the number of steps taken by the slowest ant \cite{Lenzen2013Pheromones} (the ant who has made the
least number of steps), because this alternative definition is absolutely indifferent to the
 scheduling order in the asynchronous case.

Besides time, the other quantity this work is concerned with, is the total amount of 
pheromones that are emitted during the execution of an algorithm.

\paragraph{Faults} The concept of faults manifests itself in reality, when a number of ants might stop
participating in a foraging process \cite{Langner2014FaultTolerant}, e.g., due to an ant eater performing its natural task.
To model this possibility, we assume that an unknown number
$f < k$ of ants may fail-stop at any given moment during the execution of an algorithm.
When devising fault-tolerant algorithms we require that the treasure is still guaranteed to
be found,
in spite of any such failures.

%% file: 03-fsm.tex

This section is concerned with ants modeled as mobile deterministic FSM, that are able to emit and
sense pheromones. We begin by proving a lower bound on the number of pheromones 
that must be utilized in order to conduct a successful search for a food source, 
and follow with algorithms (both synchronous and asynchronous) that achieve this bound, 
but still terminate in optimal time.

\subsection{Lower Bound}

\begin{theorem}\label{thm:fsm}
A deterministic FSM ants algorithm must emit at least $\Omega(D)$ pheromones to find a treasure located at arbitrary distance $D \in \mathcal{N}$.
\end{theorem}

\begin{proof}
Let us assume by contradiction that an FSM $\mathcal{M}$ exists which is able to find a
 treasure located at any distance $D \in \mathcal{N}$, while only emitting $o(D)$ pheromones.
Let $S$ be the number of different states in $\mathcal{M}$.

\begin{definition}[Layer]
A layer $l \in \mathcal{N}$ is the set of grid cells $(i, j) \in \mathcal{Z}^2$ at distance $l$ from
the origin, i.e., $|i| + |j| = l$.
\end{definition}

A distance $D \in \mathcal{N}$ exists, such that during $\mathcal{M}$'s search for a 
treasure located at that distance, there are $S+1$ consecutive pheromone-free layers.
 Otherwise, at least
$D/(S+1) = \Theta(D)$ pheromones are emitted for every $D \in \mathcal{N}$, in contradiction to the assumption. Let us denote these pheromone-free layers as layers $l_0, l_1, \dots, l_S$ 
(where $l_S = l_0 + S$).

Consider the first time an ant arrives at layer $l_S$, and let $s_0$ be the state in which it arrives at layer $l_S$.
If we look
at the last $S+1$ locations (the last one being at $l_S$) this ant has visited, all of them are within layers $l_0$ to
$l_S$. By the pigeon-hole principle there must be a loop, and the last state ($s_0$) is part of it.
Therefore, state $s_0$ must be reached at least once more in the last $S$ steps. Assume
w.l.o.g., that happens in layer $l_i$. So we have a path starting at a location in layer $l_i$ with state
 $s_0$ and ending
 at a location in $l_S$ with the same state. During this path the FSM does not encounter nor emit any
pheromones, and therefore this path will now repeat itself forever.

The FSM repeating this path ad infinitum will not find a treasure located outside
of the repeated pattern formed by this path, in contradiction to the original assumption.
\end{proof}

\begin{observation}
Theorem \ref{thm:fsm} holds for a single ant. Applying the same argument to a group of $k$
ants implies that unless $\Omega(D)$ pheromones are emitted, each of the $k$ ants will
eventually converge to repeating some (possibly distinct) pattern. 
Any $k$ such patterns cannot cover the whole plain.
\end{observation}

\begin{corollary}
A lower bound on the number of pheromones that must be emitted by $k$ deterministic
FSM ants, searching for a treasure located at an unknown distance $D \in \mathcal{N}$, 
is $\Omega(D)$ pheromones.
\end{corollary}

\subsection{An asynchronous FSM Ant algorithm}

\begin{definition}[Ray]
A ray of pheromones is a set of consecutive grid cells that extend from the nest in a specific
direction (north, east, south or west), such that each cell on the ray contains a pheromone. 
E.g., the northern ray contains the grid cells $(0, i) \in \mathcal{Z}^2$ such that
$i = 1, 2, ..., m$ where $(0, m)$ is the last grid cell that still contains a pheromone.
Note that the nest itself, $(0,0)$, is not considered to be part of any ray, and remains 
pheromone-free throughout the algorithm.

Similar to \cite{Emek2013FSM}, such pheromone rays are utilized as markers that enable FSM ants to perform a spiral
 shaped search, without having to count. 
\end{definition}

The asynchronous algorithm is therefore straightforward; an ant first travels along the eastern ray to
extend it by an extra pheromone, and then returns back to the nest (which is at the first
pheromone-free grid cell at the western end of the ray). 
It then repeats this process for the southern and western rays.
 Once the eastern, southern and western rays are extended, the ant travels along
the northern ray until the first pheromone-free grid cell. It emits a pheromone at that grid
cell and then starts executing zig-zag moves to the south-east (pairs of a step to the south
followed by a step to the east) until the eastern ray is encountered, at which point the
zig-zag's direction is changed to south-west, until the southern ray is encountered.
The process continues in a similar manner until the layer is fully explored and the
 ant is back at the northern ray, from
which it is able to travel back to the nest. After arriving back at the nest, the ant restarts the
process to explore another new layer. This process is repeated until the treasure is
 eventually found. See algorithm \ref{alg:async_fsm} for a more formal description, and
figure \ref{fig:async_fsm} for an illustration of the process.

\begin{figure}[!ht]
\centering
\begin{tikzpicture}[scale=0.70]
  \node[align=center] at (-4,3) {\scriptsize 3};
  \node[align=center] at (-4,2) {\scriptsize 2};
  \node[align=center] at (-4,1) {\scriptsize 1};
  \node[align=center] at (-4,0) {\scriptsize 0};
  \node[align=center] at (-4,-1) {\scriptsize -1};
  \node[align=center] at (-4,-2) {\scriptsize -2};
  \node[align=center] at (-4,-3) {\scriptsize -3};
  \node[align=center] at (-3,-4) {\scriptsize -3};
  \node[align=center] at (-2,-4) {\scriptsize -2};
  \node[align=center] at (-1,-4) {\scriptsize -1};
  \node[align=center] at (0,-4) {\scriptsize 0};
  \node[align=center] at (1,-4) {\scriptsize 1};
  \node[align=center] at (2,-4) {\scriptsize 2};
  \node[align=center] at (3,-4) {\scriptsize 3};
  \draw[thick,->] (-4,3.75) -- (-4,4.5) node[anchor=north][above] {North};
  \draw[thick,->] (3.75,-4) -- (4.5,-4) node[anchor=east][right] {East};
  \draw[thick,->] (-4,-3.75) -- (-4,-4.5) node[anchor=south][below] {South};
  \draw[thick,->] (-3.75,-4) -- (-4.5,-4) node[anchor=west][left] {West};

  \draw[dashed] (-4,3.5) -- (4,3.5);
  \draw[dashed] (-4,2.5) -- (4,2.5);
  \draw[dashed] (-4,1.5) -- (4,1.5);
  \draw[dashed] (-4,0.5) -- (4,0.5);
  \draw[dashed] (-4,-0.5) -- (4,-0.5);
  \draw[dashed] (-4,-1.5) -- (4,-1.5);
  \draw[dashed] (-4,-2.5) -- (4,-2.5);
  \draw[dashed] (-4,-3.5) -- (4,-3.5);

  \draw[dashed] (-3.5,4) -- (-3.5,-4);
  \draw[dashed] (-2.5,4) -- (-2.5,-4);
  \draw[dashed] (-1.5,4) -- (-1.5,-4);
  \draw[dashed] (-0.5,4) -- (-0.5,-4);
  \draw[dashed] (0.5,4) -- (0.5,-4);
  \draw[dashed] (1.5,4) -- (1.5,-4);
  \draw[dashed] (2.5,4) -- (2.5,-4);
  \draw[dashed] (3.5,4) -- (3.5,-4);

  \fill[gray] (0,1) circle (4pt);
  \fill[gray] (0,2) circle (4pt);
  \fill[gray] (1,0) circle (4pt);
  \fill[gray] (2,0) circle (4pt);
  \fill[gray] (0,-1) circle (4pt);
  \fill[gray] (0,-2) circle (4pt);
  \fill[gray] (-1,0) circle (4pt);
  \fill[gray] (-2,0) circle (4pt);

  \draw [thick, <->, shorten <= 1, shorten >= 1] (0,0) to node[above] {2} (2,0);
  \draw [thick, <->, shorten <= 1, shorten >= 1] (0,0) to node[right] {3} (0,-2);
  \draw [thick, <->, shorten <= 1, shorten >= 1] (0,0) to node[below] {4} (-2,0);
  \draw [thick, ->, shorten <= 1, shorten >= 1] (0,0) to node[left] {5} (0,2);

\end{tikzpicture}
\begin{tikzpicture}[scale=0.70]
  \node[align=center] at (-4,3) {\scriptsize 3};
  \node[align=center] at (-4,2) {\scriptsize 2};
  \node[align=center] at (-4,1) {\scriptsize 1};
  \node[align=center] at (-4,0) {\scriptsize 0};
  \node[align=center] at (-4,-1) {\scriptsize -1};
  \node[align=center] at (-4,-2) {\scriptsize -2};
  \node[align=center] at (-4,-3) {\scriptsize -3};
  \node[align=center] at (-3,-4) {\scriptsize -3};
  \node[align=center] at (-2,-4) {\scriptsize -2};
  \node[align=center] at (-1,-4) {\scriptsize -1};
  \node[align=center] at (0,-4) {\scriptsize 0};
  \node[align=center] at (1,-4) {\scriptsize 1};
  \node[align=center] at (2,-4) {\scriptsize 2};
  \node[align=center] at (3,-4) {\scriptsize 3};
  \draw[thick,->] (-4,3.75) -- (-4,4.5) node[anchor=north][above] {North};
  \draw[thick,->] (3.75,-4) -- (4.5,-4) node[anchor=east][right] {East};
  \draw[thick,->] (-4,-3.75) -- (-4,-4.5) node[anchor=south][below] {South};
  \draw[thick,->] (-3.75,-4) -- (-4.5,-4) node[anchor=west][left] {West};

  \draw[dashed] (-4,3.5) -- (4,3.5);
  \draw[dashed] (-4,2.5) -- (4,2.5);
  \draw[dashed] (-4,1.5) -- (4,1.5);
  \draw[dashed] (-4,0.5) -- (4,0.5);
  \draw[dashed] (-4,-0.5) -- (4,-0.5);
  \draw[dashed] (-4,-1.5) -- (4,-1.5);
  \draw[dashed] (-4,-2.5) -- (4,-2.5);
  \draw[dashed] (-4,-3.5) -- (4,-3.5);

  \draw[dashed] (-3.5,4) -- (-3.5,-4);
  \draw[dashed] (-2.5,4) -- (-2.5,-4);
  \draw[dashed] (-1.5,4) -- (-1.5,-4);
  \draw[dashed] (-0.5,4) -- (-0.5,-4);
  \draw[dashed] (0.5,4) -- (0.5,-4);
  \draw[dashed] (1.5,4) -- (1.5,-4);
  \draw[dashed] (2.5,4) -- (2.5,-4);
  \draw[dashed] (3.5,4) -- (3.5,-4);

  \fill[gray] (0,1) circle (4pt);
  \fill[gray] (0,2) circle (4pt);
  \fill[gray] (1,0) circle (4pt);
  \fill[gray] (2,0) circle (4pt);
  \fill[gray] (0,-1) circle (4pt);
  \fill[gray] (0,-2) circle (4pt);
  \fill[gray] (-1,0) circle (4pt);
  \fill[gray] (-2,0) circle (4pt);

  \draw [thick, -, shorten <= 1] (0,2) to (1,2);
  \draw [thick, ->, shorten >= 1] (1,2) to node[right] {6} (1,1);
  \draw [thick, -, shorten <= 1] (1,1) to (2,1);
  \draw [thick, ->, shorten >= 1] (2,1) to node[right] {6} (2,0);
  \draw [thick, -, shorten <= 1] (2,0) to (2,-1);
  \draw [thick, ->, shorten >= 1] (2,-1) to node[below] {7} (1,-1);
  \draw [thick, -, shorten <= 1] (1,-1) to (1,-2);
  \draw [thick, ->, shorten >= 1] (1,-2) to node[below] {7} (0,-2);
  \draw [thick, -, shorten <= 1] (0,-2) to (-1,-2);
  \draw [thick, ->, shorten >= 1] (-1,-2) to node[left] {8} (-1,-1);
  \draw [thick, -, shorten <= 1] (-1,-1) to (-2,-1);
  \draw [thick, ->, shorten >= 1] (-2,-1) to node[left] {8} (-2,-0);
  \draw [thick, -, shorten <= 1] (-2,0) to (-2,1);
  \draw [thick, ->, shorten >= 1] (-2,1) to node[above] {9} (-1,1);
  \draw [thick, -, shorten <= 1] (-1,1) to (-1,2);
  \draw [thick, ->, shorten >= 1] (-1,2) to node[above] {9} (0,2);
  \draw [thick, ->, shorten <= 1, shorten >= 1] (0,2) to node[left] {10} (0,0);

\end{tikzpicture}
\caption{Visualization of algorithm \ref{alg:async_fsm}: ray extension and exploration. Circles mark pheromones.}
\label{fig:async_fsm}
\end{figure}
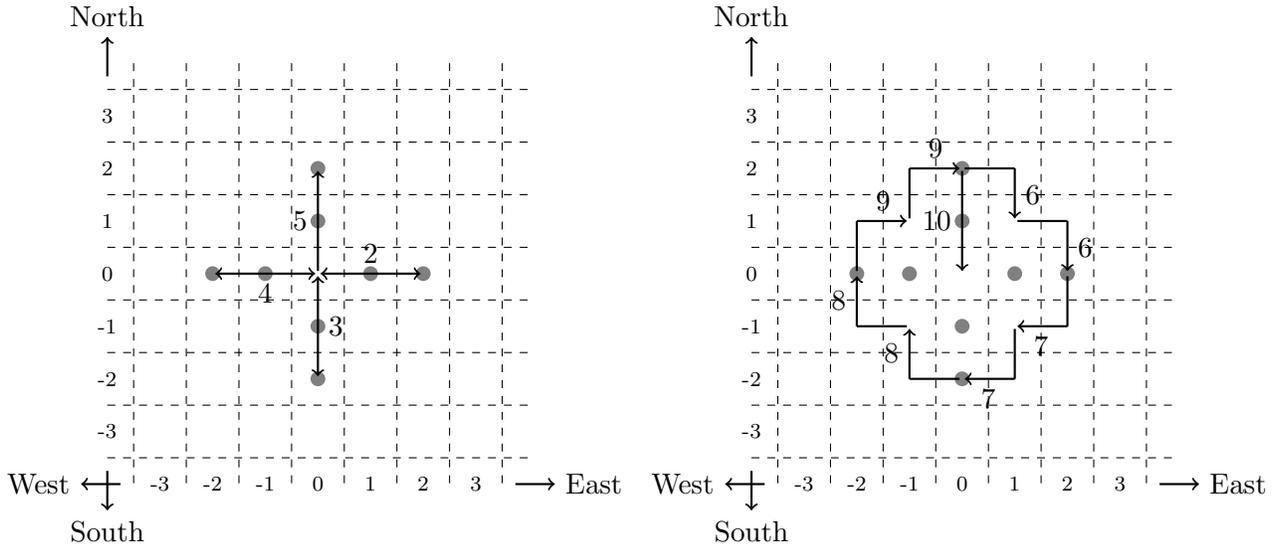

\begin{theorem}
For any $k \in \mathcal{N}$ asynchronous FSM ants and a food
 source located at an arbitrary distance $D \in \mathcal{N}$ from the nest, 
algorithm \ref{alg:async_fsm} successfully terminates in $\mathcal{O}(D + D^2/k)$
rounds, using $\mathcal{O}(D)$ pheromones.
\end{theorem}

\begin{proof}
The northern ray is never longer than the eastern, southern or western rays. Therefore,
the required pheromone guides will always exist in every explored layer.

Each layer is explored only by a single ant due to the definition of the asynchronous
 model: only one ant senses the lack of a pheromone and emits one
 on the first grid cell of each layer (which is part of the northern ray). 

Moreover, under the asynchronous model, the number of steps taken by the slowest ant
 is an upper bound on the total number of
 rounds. Therefore, we may assume that all ants move at the same pace (once per round),
which leads directly to the fact that $k$ ants explore $k$ layers in $\mathcal{O}(D)$ rounds,
or $D$ layers in $\mathcal{O}(D + D^2 / k)$ rounds.

Each layer requires a single pheromone on every ray, which is $\mathcal{O}(1)$
pheromones per layer. See remark \ref{remark:pheromones} below, which establishes that
a total of $\mathcal{O}(D)$ pheromones is sufficient.
\end{proof}

\begin{algorithm}
  \caption{Asynchronous FSM; distributed treasure search.}
  \label{alg:async_fsm}
  \begin{algorithmic}[1]
      \While{$true$}
        \State go(east), emit(), go(west) \Comment{go(dir) repeats step while sensing a pheromone.}
        \State go(south), emit(), go(north) \Comment{emit() emits a pheromone.}
        \State go(west), emit(), go(east)
        \State go(north), emit() \Comment{Four rays extended, ready to explore.}
        \State explore(east, south) \Comment{explore(zig, zag) alternates steps until a pheromone is sensed.}
        \State explore(south, west)
        \State explore(west, north)
        \State explore(north, east) \Comment{New layer explored.}
        \State go(south)\Comment{Back to nest.}
      \EndWhile
  \end{algorithmic}
\end{algorithm}

\begin{remark}\label{remark:pheromones}
There exists a schedule for algorithm \ref{alg:async_fsm},
in which the ants emit a potentially unbounded number
of pheromones (e.g., if the ant that is about to find the treasure is never scheduled to make 
that crucial step). However, we only consider the minimum \emph{required}
amount of pheromones; i.e., if the ants were not able to emit more than $\mathcal{O}(D)$
pheromones, the given algorithm would eventually still be able to find the treasure
at any location $D \in \mathcal{N}$.

At any rate, the fault-tolerant algorithms presented in section \ref{section:ft} overcome this 
issue altogether.
\end{remark}

\subsection{A synchronous FSM Ant algorithm}\label{sec:sync_fsm}

The asynchronous algorithm \ref{alg:async_fsm} under-performs in the synchronous case,
due to the new risk of collisions: new ants may leave the nest for the first time exactly at the
 wrong moment,
 so that they collide with ants passing through the nest 
after completing an exploration phase. From that point onwards, collided ants will have
no way to distinguish one from the other, and will perform the same work.
In turn, this analysis concludes a worst case running-time of $\mathcal{O}(D^2)$ rounds.

We propose a synchronous algorithm that is similar to the asynchronous one, except for its ability to 
solve the aforementioned issue. This is based on the observation, that if two ants do not collide when
 extending the eastern ray, then they never collide; the distance between any two such ants can only
 increase as the algorithm progresses. Therefore, the ants are divided to two logical groups:
newbie ants and veteran ants. An ant is considered a veteran as soon as it starts exploring its first layer,
and is a newbie until that point. Both newbie and veteran ants operate
 exactly as any ant in the asynchronous case, except for a subtle difference in the
way the eastern ray is extended.

\begin{algorithm}
  \caption{Synchronous FSM; eastern ray extension. Rest follows algorithm \ref{alg:async_fsm}.}
  \label{alg:sync_fsm}
  \begin{algorithmic}[1]
  \While{$extended$ == $false$}
    \If{$newbie$ == $true$} \Comment{Newbie ants.}
      \State step(east)
      \IIf{sense() == $false$} emit() \EndIIf \Comment{sense() checks for pheromones.}
      \State step(north) \Comment{Go north anyway to keep gaps.}
      \State idle() \Comment{Wait an extra round for veterans.}
      \IIf{sense() == $false$} emit(), $extended$ $\leftarrow$ $true$ \EndIIf
      \State step(south) \Comment{Back to ray, and to nest if extended.}
    \Else \Comment{Veteran ants.}
      \State step(east)
      \If{sense() == $false$} 
        \State emit(), $extended$ $\leftarrow$ $true$
        \State step(north), emit()
        \State step(south)
      \EndIf
    \EndIf
  \EndWhile
  \end{algorithmic}
\end{algorithm}

When veteran ants attempt to extend the eastern ray, they emit two
pheromones: a pheromone on the grid cell $(i, 0) \in \mathcal{Z}^2$, which extends the eastern ray
 (as before), and another one on the grid cell $(i, 1)$ directly to the north. 
This mechanism allows newbie ants, during their similar attempt to extend the eastern ray, to wait a
 single round for that extra northern pheromone to appear, which will indicate a possible collision
 with a veteran ant;
if a collision is detected, the newbie ant will proceed to the next pheromone-free
grid cell on the eastern ray, to try again. Note that newbie ants should perform this check
of the northern grid cells in \emph{every step} on the eastern ray --
this is in order to keep newbie ants from colliding with other newbies, by maintaining the gaps introduced by the gradual
emission process of the ants.
 If there's no collision, newbie ants will emit the second pheromone exactly as would
veterans, so that the first part of the exploration phase would always look for two consecutive pheromones
to alter its zig-zag exploration pattern.
A formal description of the key modification in the extension of the eastern ray is provided in algorithm \ref{alg:sync_fsm}.

\begin{theorem}
\label{thm:sync_fsm}
For any $k \in \mathcal{N}$ synchronous FSM ants and a food
 source located at an arbitrary distance $D \in \mathcal{N}$ from the nest, 
 algorithm \ref{alg:sync_fsm} successfully terminates in $\mathcal{O}(D + D^2/k)$
rounds, using $\mathcal{O}(D)$ pheromones.
\end{theorem}

\begin{proof}
Layer exploration phase is correct for the same reasons as before; the northern ray is never
longer than any other ray.

The most important new property that has to be proven, is that there are no collisions, i.e.,
every ant explores distinct layers. This is the case because:
(1) the gaps between newbie ants are kept, so newbies do not collide with other newbies, 
(2) newbie ants do not collide with veteran ants due to the new signaling mechanism, 
which introduces gaps between them, and
(3) veteran ants cannot collide with other veteran ants because the gaps between veterans 
can only grow: if ant $a$ starts exploring layer $l_a$ before ant $b$ starts layer $l_b$, then 
$l_a < l_b$, and $a$ will finish exploring $l_a$ before $b$ would finish $l_b$ (at which
point ant $a$ would head back to the nest, and remain ahead of ant $b$). This induction
 holds initially since every veteran starts exploring its first layer as a newbie.
Therefore, no two ants perform the same work, and the algorithm is both complete 
and correct.
\end{proof}

%% file: 04-ft.tex

This section extends the analysis of mobile deterministic FSM ants, by considering the case
where $f < k$ ants may fail-stop at any point during the search.
We propose synchronous and asynchronous fault-tolerant algorithms that consume
 $\mathcal{O}(D)$ pheromones as before, but still guarantee successful termination.
The presented algorithms execute in $\mathcal{O}(D + D^2 / (k-f) + D f)$
 rounds, which presents an additional price of $\mathcal{O}(Df)$ rounds to handle failures.

\section{An asynchronous Fault-Tolerant FSM Ant algorithm}

When ants may fail-stop, the asynchronous algorithm \ref{alg:async_fsm} may
fail to find the food source; e.g., the ant that starts exploring layer $D$, which contains
the food, may die just before finding it. In this case, no other ant will ever explore layer $D$
again. 

To recap, the aforementioned algorithm proceeds in phases, for each ant: first, there is the ray extension
phase, that extends the four rays, and then there is the exploration phase, that explores
a certain layer and relies on the pheromones emitted during the first phase. These
two phases are repeated until the food is found.
To present a fault-tolerant algorithm, we modify the
 exploration phase and introduce a verification phase right after it. 
These two modifications are described below.

After finishing an exploration phase, an ant emits an extra pheromone on the newly explored
layer, at the grid cell directly to the east of the northern ray, 
i.e., as the final step in exploring layer $l \in \mathcal{N}$, a pheromone is emitted on grid cell $(1, l)$.
 This pheromone indicates that the layer has been fully explored. Once this pheromone is
 emitted and the ant is back at the northern ray, it initiates the verification phase;
it travels along the northern ray back to the nest, and verifies that every layer on that
path is fully explored -- re-exploring every incomplete layer.
A formal description is provided in algorithm \ref{alg:async_ft}.

\begin{algorithm}
  \caption{Asynchronous Fault-Tolerant FSM; verification phase. Rest follows algorithm \ref{alg:async_fsm}.}
  \label{alg:async_ft}
  \begin{algorithmic}[1]
      \While{sense() == $true$} \Comment{Still on the northern ray, not at the nest.}
        \State step(east)
        \If{sense() == $false$} \Comment{Found unfinished layer.}
          \State explore\_layer\_fsm()  \Comment{Explore (using pheromones) and get back to northern ray.}
          \State step(east) \Comment{Mark layer as completed,}
          \State emit() \Comment{by emitting a pheromone.}
        \EndIf
        \State step(west) \Comment{Back to northern ray.}
        \State step(south)
      \EndWhile
  \end{algorithmic}
\end{algorithm}

\begin{theorem}
\label{thm:async_ft}
For any $k \in \mathcal{N}$ asynchronous FSM ants, any
$f < k$ ants fail-stopping during the algorithm, and a food
 source located at an arbitrary distance $D \in \mathcal{N}$ from the nest, 
 algorithm \ref{alg:async_ft} successfully terminates in $\mathcal{O}(D + D^2/(k-f) + Df)$
rounds, using $\mathcal{O}(D)$ pheromones.
\end{theorem}

\begin{proof}
An ant proceeds to explore a new layer after completing layer $l_{done}$, only
if it has verified that all layers $l < l_{done}$ are fully explored. 
This proves that no more than $\mathcal{O}(D + k) = \mathcal{O}(D)$ pheromones will 
be emitted\footnote{Thereby resolving remark \ref{remark:pheromones}.}, and that every
 layer $l \le D$ will eventually be explored.

To analyze the running-time, let us begin by assuming that \textbf{no failures} are possible. 
Under this assumption, consider the case where an ant $a$ that has just finished exploring layer $l_{done}$ also
re-explores layers $l_0, l_1, ..., l_n$ during its verification phase, whereas $l_n$ is the
last layer ant $a$ explores before reaching the nest. When ant $a$ started exploring
 $l_{done}$, there was already another ant $b$ at layer $l_n$ (otherwise ant $a$ wouldn't proceed).
 Since $l_n < l_{done} \le D$,
it then follows that ant $b$ has done less than $\mathcal{O}(D)$ steps by the time
that ant $a$ has reached layer $l_n$ (otherwise layer $l_n$ would have already been fully
explored), and therefore no more than $\mathcal{O}(D)$ 
rounds could have passed. Hence, to conclude: without fail-stops, every $\mathcal{O}(D)$ rounds,
 every ant explores a new layer. Overall, it takes no more than $\mathcal{O}(D + D^2/k)$ 
rounds for the $k$ ants to find the treasure in this case.

The previous analysis holds for every set of layers which has been successfully explored
without any ant failing in it. When an ant that explores layer $l \le D$ fails,
then we have, according to the
 previous analysis: after at most $\mathcal{O}(D)$ rounds all the ants will get to layer $l$,
 and after another $\mathcal{O}(D)$ rounds the layer will be fully explored. 
So, in the worst case, each fail-stop event is associated with $\mathcal{O}(D)$ extra rounds, which
 amounts to at most $\mathcal{O}(Df)$ additional rounds.
\end{proof}

\section{A synchronous Fault-Tolerant FSM Ant algorithm}

To handle the added risk of fail-stops, the synchronous algorithm \ref{alg:sync_fsm} has to
be extended in the same manner as presented in the asynchronous fault-tolerant 
algorithm
 \ref{alg:async_ft}. That is, the synchronous algorithm (presented in section \ref{sec:sync_fsm}), which uses newbie
 and veteran
ants in order to prevent collisions, and progresses in alternating phases of ray extension and
exploration, should now be updated to mark explored layers with an extra pheromone, and
has to conduct a verification phase after each exploration. 
A full description is provided in algorithm \ref{alg:sync_ft}.

\begin{algorithm}
  \caption{Synchronous Fault-Tolerant FSM; combination of algorithms \ref{alg:sync_fsm} and \ref{alg:async_ft}.}
  \label{alg:sync_ft}
  \begin{algorithmic}[1]
  \State $newbie \leftarrow true$ \Comment{Synchronous ants start as newbies.}
      \While{$true$} 
  \State $extended \leftarrow false$ \Comment{Algorithm \ref{alg:sync_fsm}; eastern ray extension.}
  \While{$extended$ == $false$}
    \If{$newbie$ == $true$} \Comment{Newbie ants.}
      \State step(east)
      \IIf{sense() == $false$} emit() \EndIIf \Comment{sense() and emit() pheromones.}
      \State step(north) \Comment{Go north anyway to keep gaps.}
      \State idle() \Comment{Wait an extra round for veterans.}
      \IIf{sense() == $false$} emit(), $extended$ $\leftarrow$ $true$ \EndIIf
      \State step(south) \Comment{Back to ray, and to nest if extended.}
    \Else \Comment{Veteran ants.}
      \State step(east)
      \If{sense() == $false$} 
        \State emit(), $extended$ $\leftarrow$ $true$
        \State step(north), emit()
        \State step(south)
      \EndIf
    \EndIf
  \EndWhile
        \State go(west) \Comment{go(dir) repeats step while sensing a pheromone.}
        \State go(south), emit(), go(north)
        \State go(west), emit(), go(east)
        \State go(north), emit() \Comment{Four rays extended, ready to explore. Now back at the northern ray.}
        \State explore\_layer\_fsm()  \Comment{Explore (zigzags with pheromones), get back to northern ray.}
        \State step(east), emit(), step(west) \Comment{Mark new layer as explored.}
        \State step(south) \Comment{Still on the northern ray.}
      \While{sense() == $true$} \Comment{Algorithm \ref{alg:async_ft}; verification phase.}
        \State step(east)
        \If{sense() == $false$} \Comment{Found unfinished layer.}
          \State explore\_layer\_fsm() 
          \State step(east) \Comment{Mark layer as completed,}
          \State emit() \Comment{by emitting a pheromone.}
        \EndIf
        \State step(west) \Comment{Back to northern ray.}
        \State step(south)
      \EndWhile
      \EndWhile
  \end{algorithmic}
\end{algorithm}

\begin{theorem}
For any $k \in \mathcal{N}$ synchronous FSM ants, any
$f < k$ ants fail-stopping during the algorithm, and a food
 source located at an arbitrary distance $D \in \mathcal{N}$ from the nest, 
algorithm \ref{alg:sync_ft} successfully terminates in $\mathcal{O}(D + D^2/(k-f) + Df)$ rounds, 
using $\mathcal{O}(D)$ pheromones.
\end{theorem}

\begin{proof}
As shown in the proof for theorem \ref{thm:async_ft}, every layer $l \le D$ is eventually
explored, and at most $\mathcal{O}(D)$ pheromones are emitted during the search.

The main new risk introduced by the synchronous model is that of collisions between ants,
i.e., ants performing the same work at the same rounds.
The proposed algorithm, in fact, does not allow such collisions to occur;
newbie and veteran ants cannot collide during the ray extension or exploration phases
 for the same reasons noted in the proof for theorem \ref{thm:sync_fsm}.
For a collision to occur during the verification phase, ant $a$ has to finish exploring a layer
 $l$, and then converge with another ant $b$ on its way to the nest along the northern ray. 
But ant $b$ checks every layer to make sure that it is explored - so it either senses the
pheromone at $(1, l)$, in which case ant $b$ is strictly after ant $a$, or it does not, in which
case ant $b$ also explores layer $l$, but ant $a$ is already done with it. 

The key point is to observe the round in which ant $a$ marks $(1, l)$ with a 
pheromone: if ant $b$ arrives to that cell strictly after or strictly before that round, then 
there is no collision. Otherwise both ants $a$ and $b$ arrive to $(1,l)$ at the same round
-- but that happens when $a$ is at the end of its exploration phase (it is about to emit
a pheromone and get back to the northern ray), and $b$ senses the lack of a pheromone
and heads to explore layer $l$.
Hence there is no collision.

Since there are no collisions, every ant explores new layers as long as there are no failures.
Every failure "wastes" at most $\mathcal{O}(D)$ rounds, which amounts to no
more than $\mathcal{O}(Df)$ extra rounds overall.
\end{proof}

%% file: 05-tm.tex

In an attempt to further reduce the number of required pheromones, a stronger 
computational model is considered. In this section, every ant has the computational
capabilities of a Turing Machine, but actually requires no more than $\mathcal{O}(\log D)$ memory bits.
We assume that there are no faults nor failures.

\section{Lower Bound}

A single TM ant can find the treasure with no pheromones in $\mathcal{O}(D^2)$ time, 
by performing a spiral shaped search.
However, we argue that pheromones must be used for the search to be cooperative, and
complete in optimal time.

\begin{theorem}
Performing a distributed search for a treasure at an unknown distance 
$D \in \mathcal{N}$ by $k$ TM ants in $O(D + D^2/k)$ rounds under the \textbf{synchronous}
 model, requires $\Omega(k)$ pheromones. 
\end{theorem}

\begin{proof}
Assume by contradiction that $T$ is a TM ant which is able to find a treasure at any
distance $D$ in $O(D+D^2/k)$ rounds with any $k$ ants, while only using $o(k)$
 pheromones.
We will devise a specific ant emission scheme, dependent on $T$, which does not achieve the
assumed running-time -- thus reaching a contradiction. 

Let us analyze the behavior of a group of $k$ ants implemented by TM $T$. We start by
 emitting only one ant at round $t_0 + 1 =1$. That ant must stop emitting pheromones at 
some round (perhaps right away),
denoted as $t_1$. $T$ only emits $o(k)$ pheromones, therefore $t_1$ must exist.
 We emit the second ant at round $t_1 + 1$,
and again wait until round $t_2$ when both ants finish emitting all of their pheromones. 
We emit the third ant
at time $t_2 + 1$ and repeat the process until we get to emit the first ant which does not
emit any pheromones: ant $p$, emitted at time $t_{p-1}+1$. There is at least one such ant,
since the total number of pheromones is $o(k)$. 
Notice that two invariants hold once ant $p$ is emitted: 
(i) no extra pheromones will be emitted
by all ants $i \le p$ according to our assumption (their view of the world does not change),
and (ii)
all ants $j > p$ operate exactly as ant $p$, since they are deterministic and share a common
 view of the world. Also worth noting is that $t_i$, for all $i \le p$, are constants which only
 depend on the specification of TM $T$.

We now increase the number of ants and keep emitting them, a single ant per round,
 until a total of $k' \gg k \ge p$
 ants have been emitted. Ants know neither $k$ nor $k'$, therefore the aforementioned behavior does not
change. However, since all $k'-p$ (at least) ants operate in the same way, there exists some 
$D \gg k'$ such that it takes the ants 
$\Theta(D + D^2/p) > O(D + D^2 / k')$ rounds to find the treasure, in
contradiction to our assumption.
\end{proof}

\begin{theorem}
Performing a distributed search for a treasure at an unknown distance 
$D \in \mathcal{N}$ by $k$ TM ants in $O(D+D^2/k)$ rounds under the \textbf{asynchronous}
 model, requires $\Omega(k)$ pheromones. 
\end{theorem}

\begin{proof}
Let us assume by contradiction that $T$ is a TM ant that finds a treasure at any
distance $D$ in $O(D+D^2/k)$ asynchronous rounds for any $k$ ants, while using
no more than $o(k)$ pheromones.

Let there be some ant group size $k$, distance $D$ and scheduling $S$, such that it takes
the $k$ ants $\Theta(D + D^2/k)$ rounds and less than $k$ pheromones to find the treasure. 
This is a tight lower bound on the running-time, therefore there must exist such a scenario. 
There is at least one ant which has not emitted any pheromones throughout the run;
let us denote one such ant as $p$.

Consider a different group of $k' \gg k$ ants, where some $k$ ants are scheduled 
exactly as in schedule $S$ and the other $k'-k$ ants are scheduled to take a single step each, 
one after
the other, right after each step of ant $p$. Ant $p$ emits no pheromones, so it shares its view of the
world with the $k'-k$ ants that follow it, and all these ants operate exactly the same (due to $T$
 being deterministic). 
The number of rounds remains $\Theta(D + D^2/k) > O(D + D^2/k')$, in contradiction to 
our assumptions.
\end{proof}

\section{Asynchronous and synchronous TM Ant algorithms}

The same algorithm works for both the asynchronous and synchronous models.
The main idea is to perform a \emph{renaming} of the ants using pheromones, followed
by a static partition of the search space, which is based on the names (IDs) of the ants.
We use pheromones to provide each ant with:
(1) a unique id $id_i \in [1, k]$, and
(2) an estimate $total_i \le k$ on the total number of participating ants.
Given these two numbers, each ant $i$ is able to compute all layers
 $l_i \equiv id_i \pmod{total_i}$, which it has to explore.

The full algorithm is as follows. Upon leaving the nest, each ant travels along the northern
ray until it encounters the first pheromone-free grid cell, $(0, i) \in \mathcal{Z}^2$. It emits
a pheromone at that grid cell and sets its $id_i = i$. At this point its estimate of the total 
number of participating ants is $total_i = i$. The ant then starts exploring layers 
$id_i, id_i + total_i, id_i + 2 total_i, ...$. After exploring every new layer, the ant returns to the 
northern ray to check for new pheromones; if it discovers that the first (closest to the nest)
 pheromone-free grid cell on the northern ray is $(0, j)$, it updates $total_i = j$ 
and restarts the search from layer $id_i + total_i$. A formal definition is provided in 
algorithm \ref{alg:tm}.

\begin{algorithm}
  \caption{Asynchronous and synchronous TM; distributed treasure search.}
  \label{alg:tm}
  \begin{algorithmic}[1]
      \State $id \leftarrow 1$, $total \leftarrow 1$, step(north)
      \While{sense() == $true$} \Comment{Searching for the first pheromone-free grid cell.}
        \State step(north)
        \State $id \leftarrow id + 1$, $total \leftarrow total + 1$
      \EndWhile
      \State emit(), $current\_layer \leftarrow id$
      \While{$true$} \Comment{Repeating exploration phase.}
        \State explore\_layer\_tm($current\_layer$) \Comment{Explore by counting steps.}
        \If{update\_total()} \Comment{If a new pheromone has appeared.}
          \State $current\_layer \leftarrow id$
        \EndIf
        \State $id \leftarrow id + total$ \Comment{Proceed to next layer.}
      \EndWhile
  \end{algorithmic}
\end{algorithm}

\begin{theorem}
For any $k \in \mathcal{N}$ asynchronous/synchronous TM ants and a food
 source located at an arbitrary distance $D \in \mathcal{N}$ from the nest, 
algorithm \ref{alg:tm} terminates in $\mathcal{O}(D + D^2/k)$
rounds, using $\mathcal{O}(k)$ pheromones.
\end{theorem}

\begin{proof}
Each ant only emits a single pheromone, which amounts to $\mathcal{O}(k)$ pheromones
overall.

For the asynchronous case, we have previously established the assumption that each ant makes
a single step in every round (otherwise, the total number of rounds may only decrease).
Therefore, each ant explores at least one new layer every $\mathcal{O}(D)$ steps, and
the search terminates in $\mathcal{O}(D + D^2/k)$ rounds.

In the synchronous case, there are no collisions because ants never leave the nest at the
 same round. Once all the ants leave the nest, the algorithm terminates in 
$\mathcal{O}(D + D^2/k)$ rounds, as it takes at most $\mathcal{O}(D)$ rounds for every
ant to update the new total (due to the assumption that $k < D$), 
and from that point each ant explores its own set of layers 
according to the static partition defined by the algorithm.
\end{proof}

%% file: 06-summary.tex

We have presented different approaches to solving the ANTS problem with
 pheromones. We have analyzed the amount of pheromones necessary and sufficient under
 different scheduling models and various computational capabilities,
and in the presence of fail-stop faults.
Table \ref{tbl:summary} summarizes our results.

\begin{table}[h]
	\centering
	\begin{tabular}{|l|c|c|c|}
		\hline
		& FSM & Fault-Tolerant FSM & TM \\ \hline
		Pheromones - Lower bound & $\Omega(D)$ & $\Omega(D)$ & $\Omega(k)$ \\ \hline
		Pheromones - Upper bound & $\mathcal{O}(D)$ & $\mathcal{O}(D)$ & $\mathcal{O}(k)$ \\ \hline
		Rounds - Lower bound & $\Omega(D + D^2/k)$ & $\Omega(D+D^2 / k)$ & $\Omega(D+D^2 / k)$ \\ \hline
		Rounds - Upper bound & $\mathcal{O}(D + D^2/k)$ & $\mathcal{O}(D + D^2/(k-f) + Df)$ & $\mathcal{O}(D + D^2/k)$ \\ \hline
	\end{tabular}
	\caption{Research summary; results hold for both asynchronous and synchronous models, whereas $D \in \mathcal{N}$ is the distance to the food, $k \in \mathcal{N}$ is the total number of participating ants, and $f < k$ is the number of fail-stops.}
	\label{tbl:summary}
\end{table}

We have presented Fault-tolerant algorithms that are resilient
to fail-stops, but still only require the computational capabilities of FSM, and no more
than the minimum amount of pheromones. Whether or not the performance or pheromone
utilization of these algorithms can be improved in the TM model, remains an open question.

Moreover, to the best of our knowledge, no lower bounds on the running-time for the 
fault-tolerant case have ever been 
presented -- the lower bound of $\Omega(D+D^2/k)$ rounds still holds, while our algorithm
terminates in $\mathcal{O}(D + D^2/(k-f) + Df)$ rounds, which leaves a gap.

Another major part of a real ant foraging process is what happens once the food (treasure) is actually
found; commonly, the discovering ant has to find its way back to the nest and inform other
 ants of the discovery in order to get assistance in carrying the food back.
This task can make use of pheromones, and as such it poses a similar problem of devising
efficient algorithms, both in terms of running-time and pheromone utilization.

Other topics for related future research could be: coping with obstacles, considering 
a similar cooperative
search originating from multiple nests, or utilizing different flavors of pheromones 
(possibly with limited lifespan).

%% file: ants.bbl
\begin{thebibliography}{10}

\bibitem{AKS2014Pheromones}
Y.~Afek, R.~Kecher, and M.~Sulamy.
\newblock {Optimal Pheromone Utilization}.
\newblock In {\em {2nd Workshop on Biological Distributed Algorithms (BDA),
  Austin, Texas, USA}}, October 2014.

\bibitem{Emek2014HowMany}
Y.~Emek, T.~Langner, D.~Stolz, J.~Uitto, and R.~Wattenhofer.
\newblock {How Many Ants Does It Take To Find the Food?}
\newblock In {\em {21th International Colloquium on Structural Information and
  Communication Complexity (SIROCCO), Hida Takayama, Japan}}, July 2014.

\bibitem{Emek2013FSM}
Y.~Emek, T.~Langner, J.~Uitto, and R.~Wattenhofer.
\newblock Ants: Mobile finite state machines.
\newblock {\em CoRR}, abs/1311.3062, 2013.

\bibitem{Emek2014FSM}
Y.~Emek, T.~Langner, J.~Uitto, and R.~Wattenhofer.
\newblock Solving the {ANTS} problem with asynchronous finite state machines.
\newblock In {\em Automata, Languages, and Programming - 41st International
  Colloquium, {ICALP} 2014, Copenhagen, Denmark, July 8-11, 2014, Proceedings,
  Part {II}}, pages 471--482, 2014.

\bibitem{Feinerman2012Memory}
O.~Feinerman and A.~Korman.
\newblock Memory lower bounds for randomized collaborative search and
  implications for biology.
\newblock In {\em Distributed Computing}, pages 61--75. Springer, 2012.

\bibitem{Feinerman2012Collaborative}
O.~Feinerman, A.~Korman, Z.~Lotker, and J.-S. Sereni.
\newblock Collaborative search on the plane without communication.
\newblock In {\em Proceedings of the 2012 ACM Symposium on Principles of
  Distributed Computing}, PODC '12, pages 77--86, New York, NY, USA, 2012. ACM.

\bibitem{Gordon2010Ants}
D.~M. Gordon.
\newblock {\em Ant encounters: interaction networks and colony behavior}.
\newblock Princeton University Press, 2010.

\bibitem{Greene2007Patrollers}
M.~J. Greene and D.~M. Gordon.
\newblock How patrollers set foraging direction in harvester ants.
\newblock {\em The American Naturalist}, 170(6):943--948, 2007.

\bibitem{Jackson2006Communication}
D.~E. Jackson and F.~L. Ratnieks.
\newblock Communication in ants.
\newblock {\em Current biology}, 16(15):R570--R574, 2006.

\bibitem{Langner2014FaultTolerant}
T.~Langner, D.~Stolz, J.~Uitto, and R.~Wattenhofer.
\newblock {Fault-Tolerant ANTS}.
\newblock In {\em {28th International Symposium on Distributed Computing
  (DISC), Austin, Texas, USA}}, October 2014.

\bibitem{Lenzen2013Pheromones}
C.~Lenzen and T.~Radeva.
\newblock {The Power of Pheromones in Ant Foraging}.
\newblock In {\em {1st Workshop on Biological Distributed Algorithms (BDA),
  Jerusalem, Israel}}, October 2013.

\end{thebibliography}
